\documentclass[3p, review]{elsarticle}
\usepackage{lineno, hyperref, graphicx, epstopdf}
\usepackage{times}
\usepackage{type1cm}
\usepackage{braket}
\usepackage{amsmath}
\usepackage{amsfonts}
\usepackage{amssymb}
\usepackage{amsthm}
\usepackage{algorithm}
\usepackage{algorithmic}
\usepackage{enumerate}
\usepackage{mathrsfs}
\usepackage{braket}
\usepackage{tikz}
\usepackage{amsmath}
\usepackage{bm}

\theoremstyle{plain}
\newtheorem{theorem}{Theorem}

\newtheorem{lemma}{Lemma}
\newtheorem{proposition}{Proposition}
\newtheorem{corollary}{Corollary}
\theoremstyle{definition}
\newtheorem{definition}{Definition}

\newtheorem{remark}{Remark}
\newtheorem{fact}{Fact}

\begin{document}
	\begin{frontmatter}
		\title{Quantum and Classical Query Complexities for Generalized Simon's Problem}
		\author{Zhenggang Wu$^{1,2}$}
		\author{Daowen Qiu$^{1, 2,}$\corref{one}}
		\author{Jiawei Tan$^{1,2}$}		
		\author{Hao Li$^{1,2}$}		
		\author{Guangya Cai$^{1}$}
		\cortext[one]{Corresponding author (D. Qiu). {\it E-mail addresses:} issqdw@mail. sysu. edu. cn (D. Qiu)}
		\address{
		$^1$ Institute of Quantum Computing and Computer Theory, School of  Computer Science and Engineering, Sun Yat-sen University, Guangzhou 510006, China \\
		$^2$The Guangdong Key Laboratory of Information Security Technology, Sun Yat-sen University, 510006, China\\}
		
		\begin{abstract}
			\emph{Simon's problem} is an essential example demonstrating the faster speed of quantum computers than classical computers for solving some problems. 
			The optimal separation between exact quantum and classical query complexities for Simon's problem has been proved by Cai $\&$ Qiu.
			Generalized Simon's problem can be described as follows. Given a function $f:{\{0, 1\}}^n \to {\{0, 1\}}^m$, 
			with the property that there is some unknown hidden subgroup
			$S$ such that $f(x)=f(y)$ iff $x \oplus y\in S$, for any $x, y\in {\{0, 1\}}^n$, where $|S|=2^k$ for some $0\leq k\leq n$. 
			The goal is to find $S$. For the case of $k=1$, it is Simon's problem. 
			In this paper, we propose an exact quantum algorithm with $O(n-k)$ queries and an non-adaptive deterministic classical algorithm with $O(k\sqrt{2^{n-k}})$ queries for solving the generalized Simon's problem. Also, we prove that their lower bounds are $\Omega(n-k)$ and $\Omega(\sqrt{k2^{n-k}})$, respectively. 
			Therefore, we obtain a tight exact quantum query complexity  $\Theta(n-k)$ and an almost tight non-adaptive classical deterministic query complexities $\Omega(\sqrt{k2^{n-k}}) \sim O(k\sqrt{2^{n-k}})$ 
			for this problem.
		\end{abstract}

	\end{frontmatter}

\section{Introduction}\label{introduction}

	The quantum query models are proven to be more powerful than their classical counterparts \cite{BW02}. 
	A {\em quantum query algorithm} is the implementation procedure of a quantum query model as follows.
	It starts with a fixed starting state $|\psi_s\rangle$ of a Hilbert ${\cal H}$ and will perform the sequence of operations $U_0, O_x, U_1,  \ldots, O_x,U_t$, 
	where $U_i$'s are unitary operators that do not depend on the input $x$, but the query $O_x$  does. This leads to the final state $ |\psi_f\rangle=U_tO_xU_{t-1}\cdots U_1O_xU_0|\psi_s\rangle$. 
	The output is obtained by measuring the final state  $ |\psi_f\rangle$.

	A quantum query algorithm ${\cal A}$ {\em exactly computes} a function $f$ if its output equals to $f(x)$ with probability 1, 
	for all input $x$. ${\cal A}$ computes $f$ {\em with bounded-error} if its output equals to $f(x)$ with probability at least $\frac{2}{3}$, for all input $x$. 
	The {\em  exact quantum query complexity} denoted by $Q_E(f)$ is the minimum number of queries used by any quantum algorithm which
	computes $f(x)$ exactly for all input $x$.

	Simon's problem conceived by Simon in 1994 \cite{Simon1994} is in the model of decision tree complexity or query complexity and it is a famous computational problem 
	that achieves exponential separation in query complexities. This problem can be defined as:
	Given a function $f:{\{0, 1\}}^n \to {\{0, 1\}}^m$, with the property that there is some unknown nonzero $s\in {\{0, 1\}}^n$ such that $f(x)=f(y)$ iff $x \oplus y\in \{0^n,s\}$, for any $x, y\in {\{0, 1\}}^n$, where zero means $0^n$. The goal is to find $s$.

	In the bounded-error setting, Simon gave an elegant quantum algorithm which solves the problem
	with $O(n)$ queries and the physical realization has demonstrated its efficiency \cite{Tame2014}.
	The $\Omega(n)$ lower bound was proved in \cite{Koiran2007} by using polynomial method \cite{Beals2001}.
	On the other hand, the classical randomized query complexity for this problem is $\Theta(\sqrt{2^{n}})$ \cite{Wolf2013}, which shows that
	the $\Theta(n)$ versus $\Theta(\sqrt{2^{n}})$ separation is an optimal one.

	For the exact query complexities of Simon's problem, Brassard $\&$ H{\o}yer \cite{Brassard1997} first gave an exact quantum algorithm solving the problem with $O(n)$ queries. 
	Then Mihara and Sung \cite{Mihara2003} proposed a simpler exact   quantum   algorithm with $O(n)$ queries in terms of a novel oracle. Recently, Cai $\&$ Qiu \cite{CQ2018} presented a straightforward exact quantum algorithm for solving Simon's problem with $O(n)$ queries.  In particular, they first gave a classical deterministic algorithm with $O(\sqrt{2^{n}})$ queries. Therefore, the optimal separation in the exact query complexities for Simon's problem is $\Theta(n)$ versus $\Theta(\sqrt{2^{n}})$.

	Moreover, Simon's problem over the general group and Simon's problem for linear functions have been studied in \cite{Alagic2008, Apeldoorn2018}. 
	Alagic $\&$ {\it al.} \cite{Alagic2008} investigated the Simon's Problem over a general group $K$, with the promise being changed, 
	and designed a quantum algorithm with time complexity $2^{O(\sqrt{n\log n})}$. 
	Apeldoorn $\&$ {\it al.} \cite{Apeldoorn2018} investigated the Simon's problem for linear functions over $\mathbb{F}_p$, 
	where $p$ is a prime power and $\mathbb{F}_p$ is a finite field with $p$ elements, and they showed the lower bound is $\Omega(n)$.

    Generalized Simon's problem proposed in \cite{KLM2007} is a generalization of Simon's problem, which can be described as follows. 
	Given a function $f:{\{0, 1\}}^n \to {\{0, 1\}}^m$, promised to satisfy the property that, 
	for some subgroup $S\subseteq {\{0, 1\}}^n$, we have, for any $x, y\in {\{0, 1\}}^n$, $f(x)=f(y)$ if and only if $x \oplus y\in S$, 
	where $|S|=2^k$ for some $0\leq k\leq n$. The goal is to find $S$. For the case of $k=1$, it is the Simon's problem.
    Computing the generalized Simon's problem   with bounded-error, 	the authors in \cite{KLM2007} gave an upper bound $O(n-k)$  on quantum query complexity with successful probability at least $\frac{2}{3}$. 
	However, we still do not know the exact quantum query complexity and classical deterministic query complexity for the generalized Simon's problem. 
	The optimal separation in exact quantum and classical deterministic query complexity for this problem needs to be clarified. 
	So, in this paper, we propose an exact quantum algorithm with $O(n-k)$ queries and a non-adaptive classical deterministic algorithm with $O(k\sqrt{2^{n-k}})$ queries for solving the generalized Simon's problem. 
	Then we show that the lower bounds on its exact quantum and non-adaptive classical deterministic query complexities are $\Omega(n-k)$ and $\Omega(\sqrt{k2^{n-k}})$ , respectively. 
	Therefore, we obtain the tight exact quantum query complexity  $\Theta(n-k)$, and the non-adaptive classical deterministic query complexities $\Omega(\sqrt{k2^{n-k}}) \sim O(k\sqrt{2^{n-k}})$ 
	for the generalized Simon's problem. When $k=1$, it accords with the results for Simon's problem obtained by Cai  and Qiu \cite{CQ2018}.

	This paper is organized as follows. In Section 2, we introduce several notations and the essential ideas for designing the exact quantum algorithms. 
	In Section 3 we present the lower bound and upper bound of exact quantum query complexity for the generalized Simon's problem. 
	Afterwards, in Section 4 we investigate the classical query complexity for this problem. Finally, conclusions are summarized in Section 5.

    
\section{Preliminaries}

	In this section, we present related definitions and notations, and give some properties of the generalized Simon's problem, 
	as well as provide the critical ideas of designing an exact quantum algorithm for the generalized Simon's problem. For more details, we can refer to  \cite{CQ2018}.

\subsection{Definitions and notations}

	Let $x, y\in{\{0, 1\}}^n$ with $x=(x_1, x_2, \cdots, x_n)$ and $y=(y_1, y_2, \cdots, y_n)$. By $x\oplus y$, we denote the bitwise exclusive-or operation, i.e.,

	\centerline{$x \oplus y=(x_1\oplus y_1, x_2 \oplus y_2, \cdots, x_n\oplus y_n)$. }

	By $x\cdot y$, we denote the inner product modulo 2 of $x$ and $y$, i.e.,

	\centerline{$x \cdot y=(x_1 y_1+x_2 y_2+\cdots+x_n y_n)$ mod 2. }

	Let $X\subseteq {\{0, 1\}}^n$. $X^{\perp}$ is a subset of ${\{0, 1\}}^n$ defined by\\
	\centerline{$X^{\perp}=\{y|\forall x\in X, x\cdot y=0\}$. }

	By $|X|$, we denote the cardinality of $X$, i.e., the number of elements of  $X$. As  \cite{CQ2018}, 
	the query set of $X$, denoted by $C_X$, is the subset of $X$, satisfying

	\centerline{$\forall x\in X, \exists y, z\in C_X, x=y\oplus z$. }

	If $X$ is a subgroup of $({\{0, 1\}}^n, \oplus)$, we denote the dimension of $X$ by $dim(X)$.

	Let $f:{\{0, 1\}}^n \to {\{0, 1\}}^m$. We use $ran(f)$ to denote the range of $f$, $dom(f)$ to denote the domain of $f$, $codom(f)$ to denote a codomain of $f$ ($codom(f)$ means some given 	set containing $ran(f)$, and here $codom(f)$  is $ {\{0, 1\}}^m$	). 
	$\mathcal{T}(f)\subseteq dom(f)$ is defined as: $ f(\mathcal{T}(f))=ran(f)$, and $ \forall x, y\in \mathcal{T}(f), f(x)\ne f(y)$.

	We use $[n]$ to denote an index set, i.e., $[n]=\{1, 2, \cdots, n\}$. 

\subsection{Generalized Simon's problem and some properties}\label{section2.2}

	The generalized Simon's problem $GSP(n,k)$ can be defined as follows:

	\textbf{Given:} $f: {\{0, 1\}}^n \to {\{0, 1\}}^m, S\subseteq {\{0, 1\}}^n, |S|=2^k, k\le n$.

	\textbf{Promise: } For all $x,y\in {\{0, 1\}}^n$, $f(x)=f(y)   \iff  (x \oplus y)\in S$.

	\textbf{Problem:} Find the hidden subgroup $S$.
    
    \begin{definition}
		Let $G={\{0, 1\}}^n$, and $\alpha_1, \cdots, \alpha_l \in G$. We call $\{\alpha_1, \cdots, \alpha_l\}$ as a linearly independent set of $G$ if and only if for any $a_1, \cdots, a_l\in\{0, 1\}$, 		$ a_1\alpha_1\oplus \cdots \oplus a_l\alpha_l=0^n\Longleftrightarrow a_1=\cdots=a_l=0$. 
	\end{definition}
	
	\begin{fact} \label{fact1}
	    $\forall s_i, s_j\in S, s_i\oplus s_j\in S$.
	\end{fact}

	\begin{fact}
		$(G, \oplus)$ is an Abelian group, and $(S, \oplus)$ is a subgroup of $(G, \oplus)$.
	\end{fact}

	\begin{lemma}\label{lemma1}
		Let $M=\{\alpha_1, \alpha_2, \cdots, \alpha_l\}$ be a linearly independent set of $G$. Denote $\braket{M}=\big\{ a_1\alpha_1\oplus \cdots \oplus a_l\alpha_l: a_1, \cdots, a_l\in\{0, 1\}\big\}$. Then $\braket{M}$ is a subgroup generated by $M$,
		and $|\braket{M}|=2^l$.
	\end{lemma}

	\begin{proof}
	    From the definition of $\braket{M}$ we can easily know that $\braket{M}$ is a subgroup of $G$.	For any ${a_1}^{(1)}\cdots {a_l}^{(1)},  {a_1}^{(2)} \cdots  {a_l}^{(2)}\in \{0, 1\}^n$ with ${a_1}^{(1)}\cdots {a_l}^{(1)}\neq {a_1}^{(2)} \cdots  {a_l}^{(2)}$, then $(a_1^{(1)}-a_1^{(2)})\alpha_1\oplus \cdots \oplus (a_l^{(1)}-a_l^{(2)})\alpha_l=0^n \Longrightarrow (a_1^{(1)}-a_1^{(2)})=\cdots=(a_l^{(1)}-a_l^{(2)})=0^n$.
		Therefore, $|\braket{M}|=2^l$.
	\end{proof}

	\begin{definition}
		Let $G_M$ be a subgroup of $G$, and let $M=\{\alpha_1, \alpha_2, \cdots, \alpha_l\}\subsetneq G_M$ be a linearly independent set of $G_M$. We define the dimension of $G_M$ equals to $|M|$, if $\braket{M}=G_M$.
	\end{definition}
	
	The following two theorems are trivial, and we omit the proofs.

	\begin{theorem}\label{theorem1}
		There exists a linearly independent set $M=\{\alpha_1, \alpha_2, \cdots, \alpha_k\}$ such that $\braket{M}=S$, and $dim(S)=|M|=k$.
	\end{theorem}
	
	\begin{theorem}
		Let $S$ be defined in the generalized Simon's problem, and $S^{\perp}=\{y: \forall s\in S, s\cdot y=0\}$. Then $S^{\perp}$ is a subgroup of $G$, $|S^{\perp}|=2^{n-k}$, and $dim(S^{\perp})=n-k$.
	\end{theorem}

	\begin{figure}[h!]
		\begin{center}
			\includegraphics[width=16cm, height=7cm]{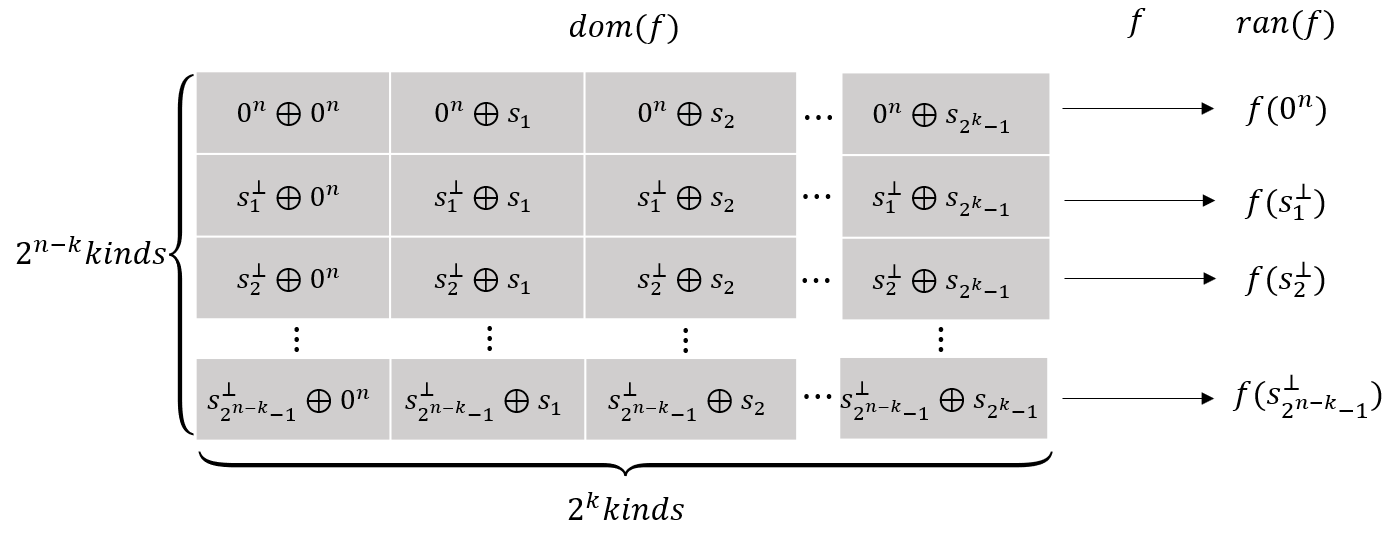}
			\caption{The illustration of $GSP(n,k)$.} \label{f1}
		\end{center}
	\end{figure}

	We provide a brief sketch of $f$ to illustrate the relation of mapping. There are precisely $2^{n-k}$ unique images for $f$, and for each element in $ran(f)$, its preimage is a set with $2^k$ elements. As Figure \ref{f1}
	shows, the left part, representing $dom(f)$, is a grid of $2^{n-k} \times 2^k$, whose elements of each rows will be mapped to a unique element in $ran(f)$,
	and $\mathcal{T}(f)$  is a subset of $dom(f)$ with $2^{n-k}$ elements selected uniquely from each row.

\subsection{Dimensional reduction}
	Dimensional reduction is  a key idea	used in whole algorithms in this paper, which uses the known results in $S$ or in $S^{\perp}$ to ensure the following result linearly independent with the previous.
	
	Brassard  and Hoyer \cite{Brassard1997} implied this idea  and came up with an exact quantum polynomial-time algorithm to solve Simon's problem. We can employ the idea to design an exact  quantum query algorithm for solving the generalized Simon's problem.

	Suppose there exists an algorithm to get a nonzero element $s\in S$ (or $z\in S^\perp$) randomly. Then we can use dimensional reduction to ensure the number of calling this algorithm can be $k$, $n-k$, the latter case etc.
	
	For $1\leq l\leq n$,
	denote $I^{(l)} \subseteq[n]=\{1, \cdots , n\}$ with $l=|I^{(l)}|$, and $I^{(l)} \subseteq I^{(l+1)}$ always hold for $l=0,1,\ldots,n-1$; 	
	denote $K^{(l)}=\lbrace x=x_{1}x_{2}\dots x_{n}: x \in \lbrace 0, 1 \rbrace^{n}, \forall j \in I^{(l)}, x_{j}=0\rbrace$ and
	$K_\perp^{(l+1)}=\lbrace x=x_{1}x_{2}\dots x_{n}: x \in K^{(l)}, \ \forall x_j \in I^{(l+1)}\backslash I^{(l)} , x_j=1\rbrace$.
	
	\begin{algorithm}
		\caption{Dimensional reduction}
		\label{alg:0}
		\begin{algorithmic}[1]
			\STATE Initial: $I^{(0)}\leftarrow\emptyset, Y=\emptyset$
			\FOR{$l\leftarrow 0:k-1$ (or $l\leftarrow 0:n-k-1$)}
			\STATE Get $s^{(l+1)} \in K^{(l)}\backslash\{0^n\}$ (or get $z^{(l+1)} \in K^{(l)}\backslash\{0^n\}$)
			\STATE Suppose $p^{(l+1)}-th$ bit of $s^{(l+1)}$ (or $z^{(l+1)}$) is nonzero, $I^{(l+1)} \leftarrow I^{(l)} \cup \{p^{(l+1)}\}$, $Y\leftarrow Y\cup \{s^{(l+1)}\}$ \\
			(or $Y\leftarrow Y\cup \{z^{(l+1)}\}$)
			\ENDFOR
			\STATE \textbf{return} $Y$
		\end{algorithmic}
	\end{algorithm}
	\begin{remark}
		$K^{(l)}$ can be divided into two parts as $K^{(l+1)}$ and ${K_\perp}^{(l+1)}$, since $\forall x \in K^{(l+1)}, x\oplus s^{(l)}\in {K_\perp}^{(l+1)}$.
	\end{remark}
	\begin{remark}
		$s^{(l+1)}, z^{(l+1)}\in K^{(l)}, s^{(l+2)}, z^{(l+2)}\notin K^{(l)}$. By induction, $\{s^{(1)}, \cdots, s^{(l)}\}$ and $\{z^{(1)}, \cdots, z^{(l)}\}$ will be two linearly independent sets of $({\{0, 1\}}^n, \oplus). $
	\end{remark}
	
	\begin{lemma}\label{lemma2}
		We have two properties as:
	
		1. $K^{(l)}\cap S=\{y=y_1y_2\cdots y_n: y\in S, \forall j \in I, y_j=0\}, |K^{(l)}\cap S|=2^{k-l}$.
	
		2. $K^{(l)}\cap S^{\perp}=\{y=y_1y_2\cdots y_n: y\in S^{\perp}, \forall j \in I, y_j=0\}, |K^{(l)}\cap S|=2^{n-k-l}$.
	\end{lemma}
	
	Now, we can draw a conclusion that the dimension of $K^{(l)}\cap S$ or $K^{(l)}\cap S^{\perp}$ will be reduced after we get a new $s\in S$ or $z\in S^{\perp}$, and then we can use this trick to keep the output set to be linearly independent for designing an exact quantum or classical algorithm, or for
	analyzing the lower bound of classical randomized algorithm.

\subsection{Quantum amplitude amplification }\label{section2.4}

	Let us recall quantum amplitude amplification\cite{Brassard2012}.

	\begin{definition}
		Let $\mathcal{A}$ be any quantum algorithm that uses no measurements, and let $\chi:\mathbb{Z}\to \{0, 1\} $ be any Boolean function.
		Assume that $\mathcal{A}\ket{0}=\ket{\Psi}=\ket{\Psi_0}+\ket{\Psi_1}$, and we call $\ket{\Psi_1}=\frac{1}{\sqrt a}\sum_{x\in A} \ket{x}$ as the good state, and $\ket{\Psi _0}=\frac{1}{\sqrt {1-a}} \sum_{x\in B} \ket{x}$ as the bad state,
		where $A\subseteq\{x\in{\{0, 1\}}^n:\chi(x)=1\}$, $B\subseteq\{x\in{\{0, 1\}}^n:\chi(x)=0\} $.

	\end{definition}
	\begin{lemma}[\cite{Brassard2012}]\label{lemma3}
		There exists a quantum algorithm that given the initial success probability $a>0$ of $\mathcal{A}$ , finds a good solution with certainty
		using a number of applications of $\mathcal{A}$ and $\mathcal{A}^{-1}$ which is in $\Theta(\frac{1}{\sqrt{a}})$ in the worst case.
	\end{lemma}

	The complementary description of Lemma \ref{lemma3} is given as follows, where $\phi$ and $\varphi$ are parameters dependent of $a$:\\
	$$ S_\chi(\varphi)\ket{x}=\left\{
		\begin{aligned}
			e^{i\varphi}\ket{x} \quad \text{if} \quad \chi(x)=1, \\
			\ket{x} \quad \text{if} \quad \chi(x)=0,             \\
		\end{aligned}
		\right.
	$$
	$$ S_0(\phi)\ket{x}=\left\{
		\begin{aligned}
			\ket{x}  \quad \text{if} \quad x=0,           \\
			e^{i\phi}\ket{x} \quad \text{if} \quad x\ne0, \\
		\end{aligned}
		\right.
	$$

	\centerline{$Q= Q(\mathcal{A}, \chi, \phi, \varphi)=-\mathcal{A}S_0(\phi)\mathcal{A}^{-1}S_\chi(\varphi). $}
	\begin{lemma}[\cite{Brassard2012}]\label{lemma4}
		Let $Q= Q(\mathcal{A}, \chi, \phi, \varphi) $. Then \\
		\centerline{$ Q\ket{\Psi_1} = e^{i\varphi}((1-e^{i\phi})a-1))\ket{\Psi_1}+e^{i\varphi}(1-e^{i\phi})a\ket{\Psi_0} $, }
		\centerline{$ Q\ket{\Psi_0} = (1-e^{i\phi})(1-a)\ket{\Psi_1}-((1-e^{i\phi})a+e^{i\phi})\ket{\Psi_0} $, } \\
		where $a=\braket{\Psi_1|\Psi_1}$.
	\end{lemma}

	\begin{corollary}\label{c1}
		There exists a quantum algorithm that given the initial success probability $\frac{1}{4} \leq a < 1 $ of $\mathcal{A}$, finds a good solution with certainty
		using applications of $\mathcal{A}$ and $\mathcal{A}^{-1}$ exactly both once. Let $\theta=\pm\arccos(1-\frac{1}{2a})$ and the specific expression of the two parameters used in
		this algorithm is given as follows: \\
		$$  \left\{
			\begin{aligned}
				& \phi=\theta+2k_1\pi, k_1 \in\mathbb{Z} ,   \\
				& \varphi=\theta+2k_2\pi, k_2 \in\mathbb{Z}.
			\end{aligned}
			\right.
		$$
	\end{corollary}

	\begin{proof} 	
		
	Since $Q$ is used once to get the good solution exactly, we have   $ Q( \ket{\Psi_1}+\ket{\Psi_0} )=(\frac{1} {\vert\vert \Psi_1 \vert\vert})   \ket{\Psi_1} $.	Therefore, 	
			by  Lemma \ref{lemma4}, the chosen $\phi$, $\varphi \in \mathbb{R}$ satisfy Eq. (\ref{e1}):
				\begin{align}
						& e^{i\varphi}(1-e^{i\phi})a=((1-e^{i\phi})a+e^{i\phi}), \text{ where} \ 0<a<1, \label{e1}        \\
			\Rightarrow & e^{i \phi}=\frac{a(e^{i\varphi}-1)}{a(e^{i\varphi}-1)+1}         \notag             , \\
			\Rightarrow & \phi=-i \log \frac{a(e^{i\varphi}-1)}{a(e^{i\varphi}-1)+1}. \notag
		\end{align}
		The definition of Logarithmic Function for complex number is shown in Eq. (\ref{e2}):
		\begin{align}
			& e^{W}=Z\Rightarrow W=\log Z=\log \lvert Z \rvert +i(\arg Z+2k\pi), k \in \mathbb{Z}. \label{e2}
		\end{align}
		Let $Z=i\phi=\frac{a(e^{i\varphi}-1)}{a(e^{i\varphi}-1)+1}$. If $\phi \in \mathbb{R}$, then $Z$ is a pure imaginary number, and $\lvert Z \rvert = 1$. Therefore, we obtain the following equations:
		\begin{align}
			\lvert Z \rvert & =-\frac{2a^{2}(\cos\varphi-1)}{2a\cos\varphi-2a+2a^{2}-2a^{2}\cos\varphi+1}=1 \label{e3}    \\
							& \Rightarrow -2a^{2}\cos\varphi+2a^{2}=2a\cos\varphi-2a+2a^{2}-2a^{2}\cos\varphi+1  \notag   \\
							& \Rightarrow 2a\cos\varphi-2a+1=0                                                    \notag  \\
							& \Rightarrow cos\varphi=1-\frac{1}{2a}                                               \notag  \\
							& \Rightarrow \varphi=\pm\arccos(1-\frac{1}{2a})+2k_1\pi, k_1\in \mathbb{Z}.            \notag
		\end{align}
	By the denominator of $Z$ being nonzero, we get the first constriction from Eq. (\ref{e3}):
		\begin{align*}
						& \cos\varphi(2a-2a^{2})-2a+2a^{2}+1\ne 0                           \\
			\Rightarrow & \cos\varphi\ne\frac{2a-2a^{2}-1}{2a-2a^{2}}=1-\frac{1}{2a-2a^{2}} \\
			\Rightarrow & 1-\frac{1}{2a}\ne 1-\frac{1}{2a-2a^{2}}                           \\
			\Rightarrow & a^{2}\ne 0 .
		\end{align*}
		By the domain of $\arccos$ defined in $[-1, 1]$, we get another constriction:
		\begin{align*}
			& -1\le1-\frac{1}{2a}\le 1      \\
			& \Rightarrow  a\ge\frac{1}{4}.
		\end{align*}
		So, we have $\varphi=\pm\arccos(1-\frac{1}{2a})+2k_1\pi, k_1\in \mathbb{Z}$, with the condition $\frac{1}{4}\le a < 1$.
		Let $\theta=\pm\arccos(1-\frac{1}{2a})$. Substitute $\varphi=\theta+2k_1\pi$ into $Z$, and then
		$Z=(1-\frac{1}{2a})+sign(\theta)\frac{\sqrt{4a-1}}{2a}i$.
		Therefore $\phi=\arg Z=\theta+2k_2\pi, k_2\in \mathbb{Z}$.
	\end{proof}

\section{Exact quantum query complexity for the generalized Simon's problem}

	In this section, we  prove that the lower bound on the exact quantum query complexity is $\Omega(n-k)$. On the other hand, an exact quantum query algorithm with $O(n-k)$ queries for solving the generalized Simon's problem is presented.

	\subsection{The lower bound}

	Koiran $et\ al$ \cite{Koiran2007} gave a lower bound on the quantum query complexity of  Simon's problem. Moreover, they transformed Simon's problem to another problem to distinguish between a trivial subgroup and a hiding subgroup, i.e., to determine whether or not the given $f$ is a bijection.
    Although the discrimination does not give the result as $s\ne 0^n$, the complexity of this transformed problem is a lower bound  on the quantum query complexity of  Simon's problem, as well.

    We utilize similar method to give a lower bound
    on the quantum query complexity for the generalized Simon's problem, but we need to change the second property of $Q_n$ in the following Proposition \ref{definition6}.  
    
    In this section, we denote by $G$ an Abelian group $({\{0, 1\}}^n, \oplus)$, and denote by $E$ the set ${\{0, 1\}}^m$.

	\begin{definition}
		Let $h:G \to E$ be a partial function, and let $f:G \to E$ be a total function. $|dom(h)|$ denotes the size of the domain of $h$, and we define:
		$$ I_h(f)=\left\{
			\begin{aligned}
				& 1 \quad \text{if $f$ extends $h$}, \\
				& 0 \quad \text{otherwise}.          \\
			\end{aligned}
			\right.
		$$
		More precisely,
		\begin{align*}
			I_h(f)=\prod \limits_{i\in domain(h), j=h(i)}{\Delta_{i, j}(f)},
		\end{align*}
	where ${\Delta_{i, j}(f)}$ is 1 if $f(i)=j$ and 0 otherwise. Then $I_h(f)$ is a monomial in the variables $(\Delta_{i, j}(f))$.

	\end{definition}

	\begin{definition}
		Let $f:G \to E$. We call $f$  hiding a subgroup $G_d$ of $G$ with order $d$, if $\forall x \in G $, $\forall y \in G_d$, $f(x)=f(x\oplus y)$.
	\end{definition}
	\begin{remark}
		For the generalized Simon's problem defined in Section \ref{section2.2}, we  have the given $f$ hiding a subgroup $S$ of $G$ with order $2^k$.
	\end{remark}

	\begin{lemma}[\cite{Beals1998}\cite{Koiran2007}] \label{lemma5}
		If $\cal A$ is an quantum algorithm of query complexity $T$, then there is a set $K$ of partial functions from $G \to E$ such that, for any function $f:G \to E$, the algorithm $\cal A$ accepts $f$ with probability
		\begin{align*}
			P_n(f)=\sum \limits_{g\in K}\alpha_{n, g}I_g(f),
		\end{align*}
		where, for every $g\in K$, we have $|dom(g)|\leq 2T(n)$ and $\alpha_{n, g}$ is real number.

	\end{lemma}

By means of \cite{Koiran2007}, we have the following proposition.

	\begin{proposition} \label{definition6}
		Suppose $\cal A$ is an algorithm computing the generalized Simon's problem with error bounded by $\epsilon\le \frac{1}{2}$. For $0\leq d \leq n $, $D=2^d$, let $Q_n(D)$ be the probability that $\cal A$ accepts $f$ when $f$ is chosen uniformly at random among the functions from $G$ to $E$ hiding a subgroup of $G$ with order $D$. If we denote by $X_D$
		the set of functions hiding a subgroup of order $D$, then we have:
		\begin{align*}
			Q_{n}(D)=\frac{1}{|X_D|}\sum_{f\in X_D}P_n(f).
		\end{align*}
		\noindent
		In addition, it has the following two properties:\\
		\noindent
		(i). for any integer $d \in[0, n]$, $0\le Q_{n}(2^d)\le1$;\\
		\noindent
		(ii). $Q_n(1)\le \epsilon$ and $Q_n(2^k)\ge 1-\epsilon$, hence $|Q'_n(x_0)|\ge \frac{1-2\epsilon}{2^k-1}>0$, for some $x_0 \in[1, 2^k]$.

	\end{proposition}
	From the above proposition, $Q_n(1)$ is the probability that $\cal A$ accepts $f$, with $f$ hiding a subgroup of $G$ of order 1, and the subgroup has only one element $0^n$. As for $Q_n(2^k)$, it represents the probability that $\cal A$ accepts $f$, with $f$ hiding a subgroup of $G$ of order $2^k$.

	We recall a useful lemma by Koiran (\cite{Koiran2007}, Lemma 5).

	\begin{lemma} [\cite{Koiran2007}] \label{Koiran2007lemma6}
	Let $c>0$ and $\xi>1$ be constants and let $P$ be a real polynomial with following properties:\\
	(i). $|P(\xi^i)|\le 1$, for any integer $0\le i \le n$;\\
	(ii). $|P'(x_0)|\ge c$, for some real number $1\le x_0 \le \xi$.
	Then $deg(P)=\Omega(n)$ , more precisely,
	\begin{align*}
	\deg(P)\ge \min\left(\frac{n}{2}, \frac{\log_2{\left(\xi^{n+3}c\right)}-1}{\log_2{\left(\frac{\xi^3}{\xi-1}\right)}+1}\right).
	\end{align*}
	\end{lemma}

	Now, we give a similar lemma as Lemma \ref{Koiran2007lemma6} above,  but change some conditions and provide a simplified proof in this section. 

	\begin{lemma} \label{lemma7}
		Let $c>0$ be a constant and let P be a real polynomial with following properties:\\
		(i). $|P(2^i)|\le1$, for any integer $0\le i \le n$;\\
		(ii). $|P'(x_0)|\ge c$, for some real number $1\le x_0 \le 2^k$.\\
	Then
		\begin{align*}
			\deg(P)\ge \min\left(\frac{n-k}{2}, \frac{n+2+\log_2{c}}{4}\right).
		\end{align*}
	\end{lemma}

	\begin{proof}  In the interest of readability, we would give  the detailed proof  here. 
		Let $d$ denote the degree of $P$. If $d\geq\frac{n-k}{2}$, the proof is complete. If $d=0$,  the second condition is not satisfied, and if $d=1$, $p'$ is a nonzero constant, so $p$ is a monotone and does not satisfy the first condition.   So, we assume $2\le d\le \frac{n-k-1}{2}$.
		
		The polynomial $P'$ and $P''$ are of degrees $d-1$ and $d-2$, respectively, so there exists an integer $a \in [n-2d+2,n-1]$ such that $P'$ has no real root in $(2^a ,2^{a+1})$, and $P''$ has no root whose real part is in this same interval. 
		It holds since there are $2d-2$ integers in this interval but these two polynomials have at most $2d-3$ real roots or real parts of root. Then, we have two properties as follows:
		
		(i) $P'$ and $P''$ are always greater than zero or always less than zero in this interval;
		
		(ii) $P$ and $P'$ are monotone in this interval.
		
		By the condition $|P(2^i)|\le 1$,  the range of $P$ in this interval $(2^a,2^{(a+1)})$ is a subset of $[-1,1]$. 
		Then we finished the first part of this proof:
		\begin{align}
			\left|P'(\frac{(2^a+2^{a+1})}{2})\right|
			&=\left||P'(\frac{3}{2}2^a)\right|         \label{eqA.1}         \\
			&\leq \max\left( \left|\frac{P(2^a)-P(\frac{3}{2}2^a)}{2^a-\frac{3}{2}2^a}\right|,\left|\frac{P(2^{a+1})-P(\frac{3}{2}2^a)}{2^{a+1}-\frac{3}{2}2^a} \right| \right)   \notag\\
			&\leq \frac{1}{2^{(a-1)}}\max\left( \left|P(2^a)-P(\frac{3}{2}2^a)\right|,\left|P(2^{a+1})-P(\frac{3}{2}2^a) \right| \right)                                   \notag     \\
			&\leq \frac{1}{2^{(a-2)}}   .                                                                                                                                    \notag
		\end{align}

		By  Equation(\ref{eqA.1}), we therefore have:
		\begin{align}
			\left|\frac{P'(\frac{3}{2}2^a)}{P'(x_0)}\right|\leq \frac{1}{c2^{a-2}}\leq \frac{1}{c2^{n-2d}}. \label{eq.A.5}
		\end{align}

		Let us write $P'(X)=\lambda \prod \limits_{i=1}^{d-1}(X-\alpha_i)$, where the $\alpha_i$s are real or complex numbers. We have the following equality:

		\begin{align}
			\left|\frac{P'(\frac{3}{2}2^a)}{P'(x_0)}\right|= \prod \limits_{i=1}^{d-1}\left| \frac{\frac{3}{2}2^a-\alpha_i}{x_0-\alpha_i} \right|. \label{eqA.2}
		\end{align}

		Let $f(x)=|\frac{\frac{3}{2}2^a-x}{x_0-x}|$. By $a\geq (n-2d+2), d\leq \frac{n-k-1}{2},1\leq x_0 \leq 2^k$, then $k-a\le k-(n-2d+2)\le -3\Longrightarrow k<a\Longrightarrow x_0<2^a$. If $x\in \mathbb{R}\backslash({x_0}\cup (2^a,2^{(a+1)})$, then $f(x)\geq \min(1, f(2^a), f(2^{a+1}))\geq \frac{1}{4}$.
		Notice that no root $\alpha_i$ of $P'$ has its real part in $(2^a,2^{(a+1)})$. Suppose $\alpha_i=\mathfrak{R}(\alpha_i)+i\mathfrak{I}(\alpha_i)$. We therefore have
		\begin{align}
			f(\mathfrak{R}(\alpha_i))\geq \frac{1}{4},   \label{eqA.3}
		\end{align}
		\begin{align*}
			\left| \frac{\frac{3}{2}2^a-\alpha_i}{x_0-\alpha_i} \right|\geq \sqrt{\frac{\left(\frac{3}{2}2^a-\mathfrak{R}(\alpha_i)\right)^2+\mathfrak{I}^2(\alpha_i)}{\left (x_0-\mathfrak{R}(\alpha_i)\right)^2+\mathfrak{I}^2(\alpha_i)}}\geq \min \left(1, \left| \frac{\frac{3}{2}2^a-\mathfrak{R}(\alpha_i)}{x_0-\mathfrak{R}(\alpha_i)} \right| \right).
		\end{align*}
		
		and thus $f(\alpha_i)\geq \frac{1}{4}$ by Equation(\ref{eqA.3}). We conclude from Equation(\ref{eqA.2}) that
		\begin{align*}
			\left|\frac{P'(\frac{3}{2}2^a)}{P'(x_0)}\right| \geq \frac{1}{4^{d-1}}.
		\end{align*}

		Taking Equation(\ref{eq.A.5}) into account, we finally obtain the following results:
		\begin{align*}
			\frac{1}{4^{d-1}}\leq \frac{1}{c2^{n-2d}}\Longrightarrow
		    d\geq \frac{n+2+\log_{2}{c}}{4}.
		\end{align*}
	\end{proof}
	
	\begin{theorem}\label{theorem3}
		If $\mathcal{A}$ is an algorithm that solves the generalized Simon's problem with bounded error  $\epsilon$ and query complexity T,
		then $T(n)=\Omega(n-k)$; more precisely,
		\begin{align*}
			T(n)\ge \min\left(\frac{n-k}{4}, \frac{n-k+3+\log_2({1-2\epsilon})}{8}\right).
		\end{align*}
	\end{theorem}
	\begin{proof}
		By the two properties of $Q_n(D)$, an application of Lemma \ref{lemma7} to polynomial $P=2Q_n(D) -1$ yields the inequality
		\begin{align*}
			\deg(P)
			& \ge \min\left(\frac{n-k}{2}, \frac{n+2+\log_2{\frac{2-4\epsilon}{2^k-1}}}{4}\right) \\
			& \ge \min\left(\frac{n-k}{2}, \frac{n-k+3+\log_2(1-2\epsilon)}{4}\right).
		\end{align*}
		Since $\deg(Q_n)\le 2T(n)$ (see, for example, \cite{Koiran2007}, Proposition 1) and $\deg(P)=\deg(Q(D))$, the proof is completed.
	\end{proof}

	Let the bounded error $\epsilon=0$ in  Theorem \ref{theorem3}. Then we can get a lower bound for quantum query complexity for the generalized Simon's problem.

	\begin{corollary}
		Any exact quantum algorithm that solves the generalized Simon's problem requires $\Omega(n-k)$ queries.
	\end{corollary}

	\subsection{The upper bound}

	Let $l\in N,0\leq l \leq n-k-1$, and let $I^{(l)}$ be an index set, which is constructed recursively by Algorithm \ref{alg:1} with an initial condition $I^{(0)}=\emptyset$.
	We use $I^{(l)}$ to construct the set $K^{(l)}$ and the quantum circuit $W^{(l)}$ as follows:

	$K^{(l)}=\lbrace x=x_{1}x_{2}\dots x_{n}, x \in \lbrace 0, 1 \rbrace^{n}: \forall j \in I^{(l)}, x_{j}=0\rbrace
	$, $W^{(l)}=\otimes_{i=1}^{n} H^{f(i)}$,
	$$  f(x)=\left\{
		\begin{aligned}
			0, i \in I^{(l)}  , \\
			1, i\notin I^{(l)}.
		\end{aligned}
		\right.
	$$

	$Q^{(l)}$ is the quantum circuit using quantum amplitude amplification to remove zero state with known amplitude (see Section \ref{section2.4}) which determines its construction.
	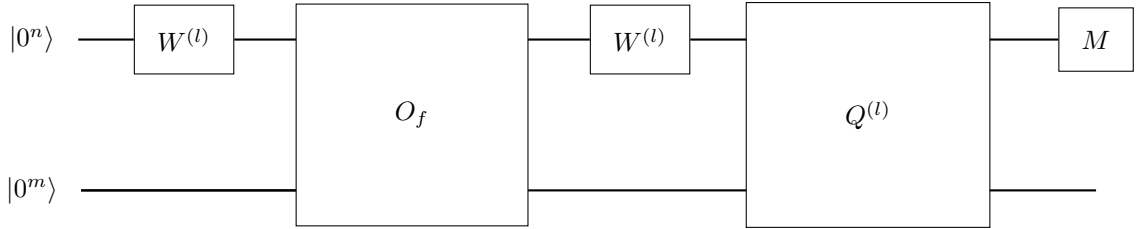
\begin{figure}[h!]\label{f2}
		\begin{center}
			\begin{tikzpicture}
				\node [rectangle, fill=white, inner sep=0. 3cm] (A) at (0, 0) {$\ket{0^{n}}$};
				\node [rectangle, fill=white, inner sep=0. 3cm] (B) at (0, -2) {$\ket{0^{m}}$};
				\node [rectangle, draw, thin, fill=white, inner sep=0. 3cm] (C) at (2, 0) {$W^{(l)}$};
				\draw[thick] (A) to (C);
				\draw[thick] (B) to (3. 5, -2);
				\node [rectangle, draw, thin, fill=white, inner sep=0. 3cm] (E) at (8, 0) {$W^{(l)}$};
				\node [rectangle, draw, thin, fill=white, inner sep=0. 3cm] (F) at (14, 0) {$M$};
				\draw[thick] (C) to (E);
				\draw[thick] (E) to (F);
				\draw[thick] (B) to (14, -2);
				\node [rectangle, draw, thin, fill=white, inner sep=1. 3cm] (D) at (5, -1) {$O_f$};
				\node [rectangle, draw, thin, fill=white, inner sep=1. 3cm] (G) at (11, -1) {$Q^{(l)}$};
			\end{tikzpicture}
			\caption{Quantum Circuit}
		\end{center}
	\end{figure}

	\begin{algorithm}
		\caption{Exact quantum algorithm for the generalized Simon's problem}
		\label{alg:1}
		\begin{algorithmic}[1]
			\STATE Initial: $I^{(0)}\leftarrow\emptyset, Y=\emptyset$
			\FOR{$l\leftarrow 0:n-k-1$}
			\STATE Prepare registers $\ket{ 0^n, 0^m}$, $W^{(l)}$, $Q^{(l)}$
			\STATE Apply $W^{(l)}$ to the first register
			\STATE Apply $O_f$ to the registers
			\STATE Apply $W^{(l)}$ to the first register
			\STATE Apply $Q^{(l)}$ to the registers    \label{a1s7}
			\STATE Measure the first register, get $z^{(l+1)} \in ( S^{\perp} \cap K^{(l)} )\backslash\{0^n\}$
			\STATE Suppose $p^{(l+1)}-th$ bit of $z^{(l+1)}$is nonzero, $I^{(l+1)} \leftarrow I^{(l)} \cup \{p^{(l+1)}\}$, $Y\leftarrow Y\cup \{z^{(l+1)}\}$
			\ENDFOR
			\STATE \textbf{return} $Y=\{z^{(1)}, \cdots , z^{(n-k)}\}$
		\end{algorithmic}
	\end{algorithm}

	\begin{theorem} \label{theorem4}
		There exists an exact quantum algorithm that solve the generalized Simon's problem with $O(n-k)$ queries.
	\end{theorem}

	\begin{proof}
		
		Let $l\in N,0\leq l \leq n-k-1$.The $lth$ loop of the algorithm is equivalent to the following formulas:
		
		(1). Prepare registers and relevant quantum circuit, the initial state is
		\begin{align*}
			\ket{0^{n}}\ket{0^{m}}
		\end{align*}
		
		(2). Apply $W^{(l)}$ to the first register
		\begin{align*}
			\xrightarrow{W^{(l)}}\quad\frac{1}{\sqrt{2^{n-l}}}\sum_{x\in K^{(l)}} \ket{x}
		\end{align*}
		
		(3). Apply $O_f$ to the registers
		\begin{align*}
			\xrightarrow{O_{f}}\quad\frac{1}{\sqrt{2^{n-l}}}\sum_{x\in K^{(l)}} \ket{x} \ket{f(x)}
		\end{align*}
		
		(4). Apply $W^{(l)}$ to the first register
		\begin{align}
			\xrightarrow{W^{(l)}}\quad & \frac{1}{2^{n-l}}\sum_{x\in K^{(l)}}\sum_{y \in K^{(l)}}(-1)^{xy}\ket{y} \ket{f(x)}                                             \label{eq4}      \\
			=                           & \frac{1}{2^{n-l}}\sum_{x \in \mathcal{T}(f)\cap K^{(l)}}\sum_{y \in K^{(l)}}[\sum_{s \in S}(-1)^{(x\oplus s) y}]  \ket{y} \ket{f(x\oplus s)} \label{eq5} \\
			=                           & \frac{1}{2^{n-l}}\sum_{x \in \mathcal{T}(f) \cap K^{(l)}}\sum_{y \in K^{(l)}}[\sum_{s \in S}(-1)^{sy}](-1)^{xy} \ket{y} \ket{f(x)}            \notag     \\
			=                           & \frac{1}{2^{n-k-l}}\sum_{x \in \mathcal{T}(f) \cap K^{(l)}}\sum_{y \in S^{\perp} \cap K^{(l)}} (-1)^{xy} \ket{y} \ket{f(x)}.       \label{eq6}
		\end{align}
	
	    In  Equation (\ref{eq5}), if $ x \in \mathcal{T}(f)$, there exist $2^{n-k}$ distinct strings mapping to $f(x)$, and these strings are in the set of
		$\{y|y=x\oplus s, s\in S\}$. Therefore, the first summation of Equation (\ref{eq4}) is divided into two parts.
		
		For Equation (\ref{eq6}), if there exists $s' \in S$ with $s' \cdot y=1$, then the value of following formula equals to zero; otherwise it will be $2^{k}$.
		\begin{align*}
			\sum_{s \in S}(-1)^{sy} & =\frac{1}{2}\sum_{s \in S}((-1)^{sy}+(-1)^{(s\oplus s')y}) \\
									& =\frac{1}{2}\sum_{s \in S}(-1)^{sy}((-1)^{s'y}+1).
		\end{align*}

		Notice that $|\mathcal{T}(f) \cap K^{(l)}|=|S^{\perp} \cap K^{(l)}|=2^{n-k-l}$, hence the first register is a uniformly superposition state that involves all the cases occurring in $ S^{\perp} \cap K^{(l)}$, and the probability of each one is $\frac{1}{2^{n-k-l}}$.
		Although we have a high probability of $a=1-\frac{1}{2^{n-k-l}}$ to get a nonzero state, there still exists some risks causing this algorithm never stops at the worst circumstance.

		By  Corollary \ref{c1}, for a given initial success probability $\frac{1}{4} \leq a < 1 $ of $Sup (f)$ and a given Boolean function  $\chi:\mathbb{Z}\to \{0, 1\} $, there exists an algorithm $Q^{(l)}$ that finds a good solution with certainty
		using applications of $\mathcal{A}$ and $\mathcal{A}^{-1}$ exactly once. Check the first condition through the following inequalities:
		\begin{align*}
			l\le n-k-1,  \   1 >a=1-\frac{1}{2^{n-k-l}}\ge \frac{1}{2}.
		\end{align*}
		
		Then, we define a  Boolean function  $\chi:\mathbb{Z}\to \{0, 1\} $ to distinguish the zero state and nonzero states:
		\begin{align*}
			\chi:{\{0, 1\}}^n \to \{0, 1\},   \    \chi(x)=0 \iff \ x=0.
		\end{align*}
		
		Therefore, the two conditions are satisfied, then the $Q^{(l)}$ used in step \ref{a1s7} can be constructed. 
		
		(5). Apply $Q^{(l)}$ to the registers
		
		Now, let us analyze the output set $Y$. After we get $z^{(l+1)}$, of which one nonzero bit $p^{(l+1)}-th$ will be added to $I^{(l+1)}$, then the next loop will output
		$z^{(l+2)}$ whose $p^{(l+1)}-th$ bit must be zero. Therefore, $z^{(l+2)}$ is linearly independent of $z^{(l+1)}$, then by induction we can draw a conclusion
		that $Y=\{z^{(1)}, \cdots , z^{(n-k)}\}$ is a linearly independent set, and $rank(Y)=n-k$. Therefore $Y$ has constructed a basis of $S^{\perp}$. Then we can calculate
		the basis of $S$ to express the whole $S$ by solving a group of linear equations.
	\end{proof}

\section{Classical query complexity for generalized Simon's problem}
    
	In this section, we show  the query complexity related to specific classical randomized algorithms. We also design a classical deterministic algorithm with $O(k\sqrt{2^{n-k}})$ queries to solve the generalized Simon's problem. We discuss a class of widely studied algorithms, i.e., non-adaptive algorithms, where each query is not allowed to depend on the result of previous queries, and derive a lower bound $\Omega (\sqrt{k2^{n-k}})$ on the non-adaptive classical deterministic query complexity. 

	\subsection{Randomized query complexity} \label{sec.4.1}

	
	
	If we have queried the oracle for $T$ times, then we say $x_1, \cdots, x_T$ are good, if there exists $ i,j\leq T,\ f(x_i)=f(x_j)$, otherwise  they are bad.
	 
	By the definition of Simon's problem, $f(x_1)\neq f(x_2) \Longleftrightarrow x_1 \oplus x_2 \neq s$. It means that the exclusive OR of pair of $x_1$ and $x_2$ can not be $s$, if they do not have the same query result. Moreover, if this classic randomized algorithm repeats $l$ queries, i.e., it queries $x_1, x_2, \cdots, x_l$, and we have $\forall i,j<l, f(x_i)\neq f(x_j) $, then we know these up to $l(l-1)/2$ pairs can not deduce $s$.
	
	If we have queried for $T-1$ times and $x_1, \cdots, x_{T-1}$ are bad, that means up to $(T-1)(T-2)/2$ pairs can not deduce $s$. Next, we query $x_T$. Then it will generate at most $T-1$ new pairs, i.e., $x_1 \oplus x_T,x_2 \oplus x_T,\cdots, x_{T-1} \oplus x_T$, and we call them potential collision.

    More specifically, the conditional probability of finding a ``collision'' in $T$ queries is as follows.
	\begin{align*}
		p(x_{1}, x_{2}, \cdots, x_{T} \text{ are good}\mid x_{1}, x_{2}, \cdots, x_{T-1}\text{ are bad})\leq \frac{T-1}{2^{n}-1-\frac{(T-1)(T-2)}{2}}\leq \frac{2T}{2^{n+1}-T^2}.
	\end{align*}

	This method can deduce a lower bound $\Omega(\sqrt{2^n})$ on classical deterministic query complexity for the Simon's problem. 
	It also needs a condition that the denominator of the above fraction is positive, i.e., $2^{n+1}-l^2>0\Longrightarrow l<\sqrt{2^{n+1}}$.
	Similarly, a trivial lower bound $\Omega(\sqrt{k2^{n-k}})$  on classical deterministic query complexity for the generalized Simon's problem can be deduced by employing this method (see Theorem \ref{theorem10}). 
    
     Given  $0\leq \varepsilon <1$, we say a randomized algorithm is successful if its  error is less than $\varepsilon$. Moreover, we say a randomized algorithm consisting of $k$ randomized sub-algorithms is failed if one of these sub-algorithms succeed with probability less than $1-\varepsilon$.
    
    \begin{theorem}\label{theorem5}
		Let $0\leq \varepsilon <1$. Then for any non-adaptive classical randomized algorithm with randomized non-repetitive input from $G={\{0,1\}}^n$,  solving the generalized Simon's problem making no more than $\sqrt{k+1}\sqrt{1-\varepsilon}\sqrt{2^{n-k}}$ queries, there exists a sub-algorithm succeeding with probability not higher than $1-\varepsilon$.
	\end{theorem}

	\begin{proof}
        
        In this proof, suppose the algorithm queries a different element on each query, until it has found a linearly independent set $M=\{s_1,\cdots,s_{k}\}$, satisfying $\braket{M}=S$. Suppose the algorithm will output these $k$ periods in $M$ step by step, using no more than $T_i$ queries for each $s_i$, $1\leq i \leq k$. Then the total queries of the algorithm is no more than $T=\sum \limits_{i=1}^k T_i$. 
        
        If the lower bound for any non-adaptive classical randomized algorithm is more than $\sqrt{2^n}$, then $\sqrt{k+1}\sqrt{1-\varepsilon}\sqrt{2^{n-k}}$ is the necessary queries in this setting, and this theorem was established. Therefore, we assume $T\leq \sqrt{2^n}$ in this proof.
		
	    Since the algorithm consists of $k$ steps, each step can be regarded as a randomized sub-algorithm. Next, we will give a specific probability analysis for each step.
		
		By non-adaptive setting, the algorithm does not use the previous information to decide the next query, so the dimensional reduction in Sec 2.3 is not used in this proof.
		
		{\bfseries Step $\bm{1}$: The probability of finding $\bm{s_{1}\neq 0 }$ using no more than $T_1$ queries}
		
		Considering the case that we have queried for $r-1$ times, where $2 \leq r \leq T_1$, but they are bad, i.e., we have not found $s_1$. There are $2^k-1$ nonzero elements in $S$, and  $\left(\sum \limits_{i=1}^{r-2}i\right)$ pairs can not deduce $s_1$. The number of potential collision are at most $r-1$. 
		
		Therefore, the successful condition probability of finding $s_1$ in the $r$-th query is as follows.
		\begin{align}
			p(x_1, \cdots , x_{r}\text{ are good}|x_1, \cdots , x_{r-1}\text{ are bad})
			\leq \frac{(2^{k}-1)(r-1)}{2^{n}-1-\sum \limits_{i=1}^{r-2}i}  .
		\end{align}
		
		The probability of finding $s_{1}$ in $S$ using no more than $T_1$ queries is
			\begin{align*}
			    P^{(1)}&=p(x_1, \cdots x_{T_1} \text{are good})\\
			    &=1-p(x_1, \cdots x_{T_1} \text{are bad})\\
			    &=1-p(x_1, \cdots x_{T_1} \text{are bad}|x_1,x_2,...,x_{T_1-1}\text{are bad})\cdot p(x_1, \cdots , x_{T_1-1}\text{ are bad})\\
			    &=1- \prod_{r=2}^{T_1}\left[p(x_1,x_2,...,x_r \text{ are bad}|x_1,x_2,...,x_{r-1}\text{ are bad})\right]\\
			     &=1- \prod_{r=2}^{T_1}\left[1-p(x_1,x_2,...,x_r \text{ are good}|x_1,x_2,...,x_{r-1}\text{ are bad})\right]\\
			     &\leq 1-\left[1-
				\sum \limits_{r=2}^{T_1}p(x_1, \cdots , x_{r}\text{ are good}|x_1, \cdots , x_{r-1}\text{ are bad}) \right]\\
			    &=
				\sum \limits_{r=2}^{T_1}p(x_1, \cdots , x_{r}\text{ are good}|x_1, \cdots , x_{r-1}\text{ are bad})\\
				&\leq\sum \limits_{r=2}^{T_1}\frac{(2^{k}-1)(r-1)}{2^{n}-1-\sum \limits_{i=1}^{r-2}i}.\\
				&\leq \frac{(2^{k}-1)\sum \limits_{r=1}^{T_1-1}r}{2^n-\sum \limits_{i=1}^{T_1-1}i}\\
				&=\frac{(2^{k}-1)\cdot T_1(T_1-1)}{2^{n+1}-T_1(T_1-1)}\\
				&\leq \frac{(2^{k}-1)\cdot T_1^2}{2^{n+1}-T_1^2}\\
			\end{align*}	       When $ T_1 \leq \sqrt{2^{n-k+1}}$, it it holds that $P^{(1)}$ is less than 1, since $(2^{k}-1)\cdot T_1^2 \leq 2^{n+1}-T_1^2$ in this case.
		  When $ T_1 \leq \sqrt{1-\varepsilon}\sqrt{2^{n-k+1}} \leq \sqrt{2^{n-k+1}}$, it holds that $P^{(1)}$ is less than $1-\varepsilon$, since
			\begin{align*}
			    p^{(1)}&\leq \frac{(2^{k}-1)\cdot T_1^2}{2^{n+1}-T_1^2}\\
			           &\overset{(a)}\leq \frac{(2^{k}-1)\cdot T_1^2+T_1^2}{2^{n+1}-T_1^2+T_1^2}\\
			           &= \frac{T_1^2}{2^{n-k+1}},
			\end{align*}
			where step(a) uses the inequality scaling as follows:
			\begin{align*}
				\frac{a}{b}\leq \frac{a+c}{b+c} \text{, for any }0\leq a\leq b, c\geq 0 \text{ and } b\neq 0. \label{eq.13}
			\end{align*}
			
			Therefore, we have shown that $\sqrt{1-\varepsilon}\sqrt{2^{n-k+1}}$ queries are necessary to attain the $s_1\in S$ with probability $1-\varepsilon$, and it is also the necessary condition of $P^{(2)}$.
		
		{\bfseries Step $\bm{l}$: The conditional probability of finding $s_{l}$ using no more than $T_l$ queries, $\bm{2\leq l\leq k}$.}
		
		We assume that, before step $l$, we have found $M_{l-1}=\{s_1, \cdots, s_{l-1} \}\subseteq S$ in previous $l-1$ steps. More specifically, we have found $s_i$ in no more than $T_i$ queries, where $1 \leq i \leq l-1$. $M_{l-1}$ is a basis of subgroup of $S$, i.e., $\braket{M_{l-1}}\subseteq S$, and $|\braket{M_{l-1}}|=2^{l-1}$. 
		
		The total queries before step $l$ is no more than $T_1+ \cdots T_{l-1}$. We denote these queries as a set $\mathcal{F}$, and then $|\mathcal{F}|\leq T_1+ \cdots T_{l-1}$. We say $x_1^{(l)},x_2^{(l)},\cdots, x_{T_l}^{(l)},\mathcal{F} $ are good, if these queries can deduce a new $s_l\in S\backslash \braket{M_{l-1}}$. 
		
		Considering the case that we have queried for $r-1$ times in step $l$, where $1 \leq r \leq T_l$, but they are bad. There are $2^k-2^{l-1}$ elements in $S\backslash \braket{M_{l-1}}$, and $\left(\sum \limits_{i=1}^{|\mathcal{F}|+r-2}i\right)$ pairs can not deduce $s_l\in S\backslash \braket{M_{l-1}}$. The number of potential collision are at most $|\mathcal{F}|+r-1$. Therefore, in the non-adaptive setting, the successful condition probability of finding $s_l$ in the $r$-th  query of step $l$ is as follows.
		
		\begin{align*}
			&p\left(x_1^{(l)},x_2^{(l)},...,x_r^{(l)},\mathcal{F} \text{ are bad}|x_1^{(l)},x_2^{(l)},...,x_{r-1}^{(l)},\mathcal{F}\text{ are bad}\right)\\
			&\leq \frac{\left(2^{k}-2^{(l-1)}\right)(|\mathcal{F}|+r-1)}{2^{n}-1-\sum \limits_{i=1}^{|\mathcal{F}|+r-2}i}  .\\
			&\leq \frac{\left(2^{k}-2^{(l-1)}\right)(T_1+ \cdots T_{l-1}+r-1)}{2^{n}-1-\sum \limits_{i=1}^{T_1+ \cdots T_{l-1}+r-2}i}  .
		\end{align*}
		Then, we can give the whole conditional probability.
		\begin{align}
			\notag
			P^{(l)}
			&=p(x_1^{(l)},x_2^{(l)},\cdots, x_{T_l}^{(l)}, \mathcal{F}|\mathcal{F} \text{ are bad})\\
			\notag
			&\leq 1- \prod_{r=1}^{T_1}\left[p\left(x_1^{(l)},x_2^{(l)},...,x_r^{(l)},\mathcal{F} \text{ are bad}|x_1^{(l)},x_2^{(l)},...,x_{r-1}^{(l)},\mathcal{F}\text{ are bad}\right)\right]\\
			\notag
			&\leq \sum \limits_{r=2}^{T_1}p\left(x_1^{(l)},x_2^{(l)},...,x_{r}^{(l)},\mathcal{F}\text{ are good}|x_1^{(l)},x_2^{(l)},...,x_{r-1}^{(l)},\mathcal{F}\text{ are bad}\right)\\
			\notag
			&\leq\sum \limits_{r=1}^{T_l}\frac{(2^{k}-2^{l-1})(T_1+\cdots+T_{l-1}+r-1)}{2^{n}-1-\sum \limits_{i=1}^{(T_1+\cdots+T_{l-1})+r-2}i} \\
			\notag
			&\overset{(a)}\leq \frac{(2^{k}-2^{l-1})\sum \limits_{r=T_1+\cdots+T_{l-1}}^{T_1+\cdots+T_l-1}r}{2^n-\sum \limits_{i=1}^{T_1+\cdots+T_l-1}i} \\
            \notag
			&= \frac{(2^{k}-2^{l})\left[2(T_1+\cdots+T_{l-1})T_l+T_l(T_l-1)\right]}{2^{n+1}-{(T_1+\cdots+T_l)}^2+{(T_1+\cdots+T_l)}}   \\
			\notag
			&\leq \frac{2^{k}\left[2(T_1+\cdots+T_{l-1})T_l+{T_l}^2\right]}{2^{n+1}-{(T_1+\cdots+T_{l})}^2}         \\
			\notag
			&= \frac{1}{2^{n-k+1}}\frac{1}{1-2^{-(n+1)}\cdot{(T_1+\cdots+T_{l})}^2}\left[2(T_1+\cdots+T_{l-1})T_l+{T_l}^2\right] \\
			&\overset{(b)}\leq \frac{1}{ 2^{n-k+1}}\left[4(T_1+\cdots+T_{l-1})T_l+2{T_l}^2\right], \notag 
		\end{align}
		 where step (a) holds since $r \le (T_l +1)$; 
			step (b) holds by means of
			\begin{align*}
			    1\geq 1-2^{-(n+1)}\cdot{(T_1+\cdots+T_{l})}^2 \geq 1-2^{-(n+1)}\cdot{T}^2\geq 1-2^{-(n+1)}\cdot 2^n = \frac{1}{2}. \notag 
			\end{align*}
     When $ T_l \leq \left(\sqrt{\beta^2+\frac{1}{2}(1-\varepsilon)\cdot 2^{n-k+1}}-\beta \right)$, it holds that $P^{(l)}$ is less than $1-\varepsilon$, where $\beta= T_1+\cdots+T_{l-1}$.
		
			Based on the previous condition that $ T_1 \geq \sqrt{1-\varepsilon}\sqrt{2^{n-k+1}}$ , we can calculate the necessary extra queries to attain $s_2\in S$ that is linearly independent of $s_1$ by the conditional probability as follows:
			\begin{align*}
				P^{(2)}
			    &\leq \frac{1}{2^{n-k+1}}(4\sqrt{1-\varepsilon}\sqrt{2^{n-k+1}}T_2+2T_2^2) \\
				&\leq (1-\varepsilon)\left[4\frac{T_2}{\sqrt{1-\varepsilon}\sqrt{2^{n-k+1}}}+2\left(\frac{T_2}{\sqrt{1-\varepsilon}\sqrt{2^{n-k+1}}}\right)^2\right].
			\end{align*}
			When $ T_2 \leq (\sqrt{1+\frac{1}{2}}-1)\sqrt{1-\varepsilon}\sqrt{2^{n-k+1}}$, it holds that  $P^{(2)}$ is less than $1-\varepsilon$.
			We can use the similar methods to prove the following results:
			
			(i) For any $2\leq i \leq l$, the number of necessary queries of $T_l$ is not less than $(\sqrt{1+\frac{l-1}{2}}-\sqrt{1+\frac{l-2}{2}})\sqrt{1-\varepsilon}\sqrt{2^{n-k+1}}$, if $P^{(i)}$ is no less than $1-\varepsilon$.
	
			(ii) The number of necessary queries of successfully finding a basis of $S$ is
			\begin{align*}
				T=\sum \limits_{i=1}^k T_i\geq \sqrt{\frac{k+1}{2}}\sqrt{1-\varepsilon}\sqrt{2^{n-k+1}}=\sqrt{k+1}\sqrt{1-\varepsilon}\sqrt{2^{n-k}}.
			\end{align*}
			
			(iii) There exist a sub-algorithm whose probability of error is more than $\varepsilon$, if the number of queries is less than $\sqrt{k+1}\sqrt{1-\varepsilon}\sqrt{2^{n-k}}$.

	\end{proof}
	\begin{remark}
	    For the adaptive setting, the adaptive algorithm can use the previous queries to determine which element to query next. In addition, such algorithms can exclude the number of $s \in S$ up to $\tbinom{r}{2}$ since previous $r$ bad queries can produce  $\tbinom{r}{2}$ collision pairs at most. However, there are some extra $s$ that can be eliminated implicitly by the fact that $S$ is a subgroup of rank $k$, when $k > 1$. 
	\end{remark}

\subsection{Deterministic query complexity}

	In this subsection the exact upper bound is given in terms of a classical deterministic algorithm. The core idea of this algorithm is to construct several query sets to minimize the number of queries and to cover the given search space, and for any input, this algorithm can find a period before it queries all element of those query sets.
	The definition of query set and the method of construction are given as follows.
	\begin{definition}
		Let ${\mathcal{B}}={\{0, 1\}}^n$. $C_\mathcal{B}\subseteq \mathcal{B}$, and it satisfies:  $ \forall x\in \mathcal{B}, \exists y, z\in C_\mathcal{B}, x=y\oplus z$.
		We call ${\mathcal{B}}$ as a search space, and $C_\mathcal{B}$ as a query set.
	\end{definition}

	\begin{theorem} \label{theorem7}
		The cardinality of $C_\mathcal{B}$ is $\Theta(\sqrt{2^n})$.
	\end{theorem}

	\begin{proof}
		We give a method to construct a query set to prove the upper bound.

		Let $I_{front}=\{1, 2, \cdots , \lfloor{n/2\rfloor}\}$, $I_{back}=\{\lfloor{n/2\rfloor}+1, \cdots , n\}$, $\mathcal{B}_{front}=\{b=b_1\cdots b_n: \forall j\in I_{front}, b_j=0\}$,
		$\mathcal{B}_{back}=\{b=b_1\cdots b_n: \forall j\in I_{back}, b_j=0\}$. Then $C_\mathcal{B}=\mathcal{B}_{front} \cup \mathcal{B}_{back}$ is a query set for $\mathcal{B}$, and
		\begin{align*}
			|C_\mathcal{B}| = 2^{\lfloor {\frac{n}{2}\rfloor}}+2^{\lceil{\frac{n}{2}\rceil}}-1\leq 2\sqrt{2^{n+1}}.
		\end{align*}
		Now consider the lower bound. Suppose $|C_\mathcal{B}|=T$.
		Then the query set can cover up to $(T-1)T/2$ elements, i.e., $T^2\geq(T-1)T\geq 2|\mathcal{B}|=2^{n+1} \Rightarrow T\geq \sqrt{2^{n+1}}$ is the necessary number of queries.

		Therefore, the cardinality of $C_\mathcal{B}$ is $\Theta(\sqrt{2^n})$.
	\end{proof}
	\begin{theorem} \label{theorem8}
		Let $G=({\{0,1\}}^n,\oplus)$ be an Abelian group, and let $S, G_s$ be the subgroups of $G$,
		where $dim(S)=k,dim(G_s)=n-k+1,dim(G)=n$. Then $(G_{s}\cap S)\backslash\{0\}\neq \varnothing$.
	\end{theorem}
	\begin{proof}
		Suppose  there exist two bases $M=\{\alpha_1,\cdots,\alpha_k\},N=\{\beta_1,\dots,\beta_{n-k+1}\}$ for $S$, $G_s$ respectively. $|M\cup N|=n+1>dim(G)=n$, and then $M\cup N$ is a linearly dependent set satisfying
		$\exists a_1,\cdots, a_k,b_1,\cdots, b_{n-k+1}\in\{0,1\}$ with not all equal to $0$ such that $a_1\alpha_1 \oplus\cdots\oplus a_k \alpha_k=b_1\beta_1 \oplus\cdots\oplus b_k \beta_{n-k+1}$.
	\end{proof}

	The trivial method to find the basis of $S$ is to construct a query set to cover $G$, where $|G|=2^n$, and a loose upper bound is given as $O(\sqrt{2^n})$. Benefited from Theorem \ref{theorem8} and the dimensional reduction, we can get a relatively tight upper bound.
	A general comprehension of Theorem \ref{theorem8} can be described as follows: for any subgroup $G_s$ of $G$ whose dimension is $n-k+1$,  then $G_s\cap S$ has at least a nonzero element.

	Therefore, once we use a query set to cover a subgroup of $G$, whose dimension is $n-k+1$, we can get at least one nonzero period $s\in S$.
	Next, $k$ different subgroups of $G$ whose dimension is $n-k+1$ can generate $k$ nonzero periods, and we use the core idea of dimensional reduction to ensure these $k$ periods are linearly independent, and they can be constructed as a basis of $S$.

		\begin{algorithm}
		\caption{Classical deterministic algorithm for the generalized Simon's problem}
		\label{alg:2}
		\begin{algorithmic}[1]
			\STATE $I^{(1)}\leftarrow [k-1], Y\leftarrow \varnothing$.
			\FOR{$i\leftarrow1:k$}
			\STATE ${\mathcal{B}}^{(i)}\leftarrow\{x=x_1\cdots x_n|x\in{\{0, 1\}}^n, \forall j \in I^{(i)}, x_j=0\}$.
			\STATE Prepare ${\mathcal{C}_\mathcal{B}}^{(i)}$, $|{\mathcal{C}_\mathcal{B}}^{(i)}|=O(\sqrt{2^{n-k+1}})$. \\
			\STATE Find the period $s_i \ne 0$ before query all elements of ${\mathcal{C}}^{(i)}$, and suppose $p^{(i)}-th$ bit of $s_i$is nonzero.\\
			\STATE $Y\leftarrow Y\cup \{s_i\}$.\\
			\IF{$i<k$}
			\STATE $I^{(i+1)}\leftarrow I^{(i)}\cup \{p^{(i)}\} \backslash \{k-i\}$.
			\ENDIF
			\ENDFOR
			\STATE \textbf{return} $Y=\{s_1, s_2, \cdots , s_{k}\}$
		\end{algorithmic}
	\end{algorithm}

	\begin{theorem} \label{theorem9}
		There exists a classical deterministic algorithm that solves the generalized Simon's problem with $O(k \sqrt{2^{n-k}})$ queries.
	\end{theorem}

	\begin{proof}

		Consider the initial condition of $i-th$ step of algorithm\ref{alg:2}:\\
		(1). $I^{(i)}= \{p^{(1)}, \cdots , p^{(i-1)}\}\cup [k-i-1]$. \\
		(2). ${\mathcal{B}}^{(i)}=\{x=x_1\cdots x_n|x\in{\{0, 1\}}^n, \forall j \in I^{(i)}, x_j=0\}$. \\
		(3). $|I^{(i)}|=k-1$.
		
		$\mathcal{B}^{(i)}$ is a subgroup of $G$, where $dim(\mathcal{B}^{(i)})=k+1$. By Theorem \ref{theorem8}, we definitely find the period $s_i \ne 0$ before we query all element of ${\mathcal{C}}^{(i)}$, and we will get an $s_i \in {\mathcal{B}}^{(i)}$.
		Let $s_i=y_1\cdots y_n$, with $y_j\in\{0,1\}$, for any $j\leq n$. By $s_i \in \mathcal{B}^{(i)}$, then $\forall t \in I^{(i)}, y_t=0$.

		Next, we add $p^{(i)}$ to ${I}^{(i+1)}$ to insure $s_i$ linearly independent of the next periods of $s_t ,t>i$, where $p{(i)}-th$ bit of $s_i$ is nonzero.
		Repeat these procedure until $i=k$, then get $Y=\{s_1, s_2, \cdots , s_{k}\}$, and it is not difficult to check that $Y$ is a maximum linearly
		independent group, which  also consists of a basis of $S$.
		
		Now, consider the cardinality of query set in this algorithm. A trivial proof of upper bound is as follows:
		\begin{align*}
			|{\mathcal{C}_\mathcal{B}}^{(i)}|=2^{\lfloor {\frac{n-k+1}{2}\rfloor}}+2^{\lceil{\frac{n-k+1}{2}\rceil}}-1\leq 2\sqrt{2^{n-k+2}}\Longrightarrow |\bigcup \limits_{i=1}^k{\mathcal{C}_\mathcal{B}}^{(i)}|\leq 2k \sqrt{2^{n-k+2}}.
		\end{align*}
		For getting the tighter upper bound we need to consider the construction of  $\mathcal{C}_{\mathcal{B}}^{(1)},\cdots,\mathcal{C}_{\mathcal{B}}^{(k)}$, satisfying	\begin{align*}
			|\bigcup \limits_{i=1}^k{\mathcal{C}_\mathcal{B}}^{(i)}|\leq 2^{\lfloor {\frac{n-k+1}{2}\rfloor}}+2^{\lceil{\frac{n-k+1}{2}\rceil}}-1+(k-1)2^{\lceil{\frac{n-k+1}{2}\rceil}-1}\leq (k+2)\sqrt{2^{n-k+1}}.
		\end{align*}

	In the interest of readability, we give  the detailed steps for the construction  in the following.

		In the worst situation, we can only get one period in one loop of this algorithm, which means that for $i\leq k$,
		only one element $p^{(i)}$ is add to $I^{(i+1)}$, and then $I^{(i+1)}=I^{(i)}\cup p^{(1)} \backslash \{k-i\}$.
		
		There gives a method to construct query set in the proof of Theorem \ref{theorem7}, and then we can construct two similar parts
		$\mathcal{B}^{(i)}_{front}$ and $\mathcal{B}^{(i)}_{back}$
		such that $C_\mathcal{B}^{(i)}=\mathcal{B}^{(i)}_{front} \cup \mathcal{B}^{(i)}_{back}$ to cover $\mathcal{B}^{(i)}$,
		where
		\begin{align*}
			\mathcal{B}^{(i)}_{front}=\{b=b_1\cdots b_n|\forall j\in I^{(i)}_{front}, b_j=0\},
			\mathcal{B}^{(i)}_{back}=\{b=b_1\cdots b_n|\forall j\in I^{(i)}_{back}, b_j=0\}.
		\end{align*}
		The construction of the two parts depends on the two sets $I^{(i)}_{front}$ and $I^{(i)}_{back}$, where
		\begin{align*}
		I^{(i)}=I^{(i)}_{front} \cup I^{(i)}_{back},
		|I^{(i)}_{front}|=\lfloor\frac{n-k+1}{2}\rfloor,
		I^{(i)}_{back}=\lceil\frac{n-k+1}{2}\rceil.
		\end{align*}
		
		If $k-i\in I^{(i)}_{front}$, then
		\begin{align*}
			I^{(i+1)}_{front}=I^{(i)}_{front}\cup p^{(i)}\backslash \{k-i\}, I^{(i+1)}_{back}=I^{(i)}_{back}.
		\end{align*}
		
		The corresponding set $\mathcal{B}^{(i)}_{back}$ can be reused as $\mathcal{B}^{(i+1)}_{back}$. Notice that
		\begin{align*}
			\mathcal{B}^{(i)}_{front} \cap \mathcal{B}^{(i+1)}_{front}=\{b=b_1\cdots b_n|\forall j\in I^{(i)}_{front}\backslash\{k-i\}, b_j=0\},
		\end{align*}
		\begin{align*}
			|C_\mathcal{B}^{(i+1)}\backslash C_\mathcal{B}^{(i)}|=|\mathcal{B}^{(i+1)}_{front} \backslash \mathcal{B}^{(i)}_{front}|=|\mathcal{B}^{(i)}_{front} \cap \mathcal{B}^{(i+1)}_{front}| =2^{\lfloor\frac{n-k+1}{2}\rfloor-1}.
		\end{align*}
		
		In a similar way, $|C_\mathcal{B}^{(i+1)}\backslash C_\mathcal{B}^{(i)}|=2^{\lceil\frac{n-k+1}{2}\rceil-1}$, if $k-i\in I^{(i)}_{back}$.Therefore
		\begin{align*}
		\forall j<l,|C_\mathcal{B}^{(j+1)}\backslash C_\mathcal{B}^{(j)}|\leq 2^{\lceil\frac{n-k+1}{2}\rceil-1}\leq 2^{\frac{n-k+1}{2}},
		\end{align*}
		
		and we get the result as follow:
		\begin{align*}
			|\bigcup \limits_{i=1}^k{\mathcal{C}_\mathcal{B}}^{(i)}|\leq 2^{\lfloor {\frac{n-k+1}{2}\rfloor}}+2^{\lceil{\frac{n-k+1}{2}\rceil}}-1+(k-1)2^{\lceil{\frac{n-k+1}{2}\rceil}-1}\leq (k+2)\sqrt{2^{n-k+1}}.
		\end{align*}
	\end{proof}    
		
	\begin{theorem} \label{theorem10}
		Any non-adaptive classical deterministic algorithm that solves the generalized Simon's problem requires $\Omega(\sqrt{k2^{n-k}})$ queries.
	\end{theorem}
    \begin{proof}
        By theorem \ref{theorem8}, for any subgroup $G_s$ of rank $n-k+1$,  there exists at least one element of $S$ in this subgroup,
        and in the worst situation there exists only one.
         The classical deterministic algorithm is successful if the $k$ periods of $S$ found from these subgroups are linearly independent.
          Otherwise, it needs extra queries. In order to get $k$ periods of $S$ to form a basis, 
        the necessary range that query set needs to cover is at least $k$ different subgroups, 
        which can be denoted by $G_1,\cdots, G_k$, and generate $s_i\in G_i$ for any $i\leq k$.
        The queries are minimal if these periods are linearly independent.
        
        Then in this case, there exists a set $M=\{\alpha_1,\cdots,\alpha_{n-k}\}$ such that $M\cup\{s_i\}$ is a basis of $G_i$, $\braket{M\cup\{s_i\}}=G_i$, for any $i\leq k$. Moreover, $\braket{M\cup \{s_1,\cdots, s_k\}}=G$.
        
        For any $i,j\leq k, i\neq j$, $s_i,s_j$ are linearly independent, and then $|G_i\backslash G_j|\geq |N_i|=2^{n-k}$, where
        \begin{align*}
            N_i=\left\{x\big|x=a_1\alpha_1\oplus\cdots\oplus a_{n-k}\alpha_{n-k}\oplus s_i, \text{for any } a_1,\cdots,a_{n-k}\in\{0,1\}\right\}.
        \end{align*}
    
        Therefore, we get the minimal queries as follows:
        \begin{align*}
            |\bigcup \limits_{i=1}^k{G_i}|\geq \sum \limits_{i=1}^{k}|N_i|=k2^{n-k}.
        \end{align*}
        
        Suppose $|C_\mathcal{B}|=T$, and then the query set can cover $(T-1)T/2$ elements at most. Furthermore, the query set needs to cover $\bigcup \limits_{i=1}^k{G_i}$, and their cardinality is at least $k2^{n-k}$.
        So, it is  required that $T^2\geq(T-1)T\geq 2|\bigcup \limits_{i=1}^k{G_i}|\geq k2^{n-k+1}$. That is to say,   the necessary number $T$ of queries  satisfies $ T\geq \sqrt{k2^{n-k+1}}$ .
          \end{proof}
		The non-adaptive classical deterministic query complexity for the generalized Simon's problem is
		$\Omega(\sqrt{k2^{n-k}})\sim O(k\sqrt{2^{n-k}})$. Hence the optimal construction that can attain this lower bound has still not been solved. So, it remains open for getting the optimal non-adaptive classical deterministic query complexity of the generalized Simon's problem.

\section{Conclusions}\label{Sec6}

	Simon's problem is a computational problem that can be solved exponentially faster on a quantum computer than on a classical  computer  \cite{Simon1994, Simon1997}. 
	The algorithm for this problem was also an  inspiration for Shor's algorithm \cite{Shor1994, Shor1997}. 
	The optimal separation between  the exact quantum query complexity and classical deterministic query complexity for Simon's problem was proved in \cite{CQ2018}. 
	The generalized Simon's problem was proposed in \cite {KLM2007}, but the optimal exact quantum query complexity and classical deterministic query complexity for the generalized Simon's problem were not clear. 
	So, in this paper, we have  tried to obtain a number of results related to these problems.

	More specifically,  we have given an exact quantum algorithm for solving this problem with $O(n-k)$ queries, and we have also shown that the lower bound on its exact quantum  query complexity is $\Omega(n-k)$. 
	Therefore, we have obtained the optimal exact quantum query complexity  $\Theta(n-k)$ for the generalized Simon's problem.
	
	For the classical complexity, we have given a non-adaptive classical deterministic algorithm with $O(k\sqrt{2^{n-k+1}})$  queries for solving the generalized Simon's problem. 
	Furthermore, we have shown that the lower bound on its non-adaptive classical deterministic query complexity  is $\Omega(\sqrt{k2^{n-k}})$. 
	
	Therefore,  the optimal non-adaptive classical deterministic query complexity for  the generalized Simon's problem is still to be solved further in the future. Another problem is to study the generalized Simon's problem with	
	 adaptive algorithms.
	
\section*{Acknowledgements}
	The authors  would like to thank the anonymous referee for important comments and suggestions that help us improve the quality of the manuscript. Also, we  would thank  Koiran for useful comments on the lower bound on the quantum query complexity of the generalized Simon's problem. This work  is  partly supported by the National Natural Science Foundation of China (Nos. 61572532,  61876195), the Natural Science
	Foundation of Guangdong Province of China (No. 2017B030311011).

\begin{appendix}

\end{appendix}

\end{document}